\documentclass[a4paper,UKenglish,cleveref, autoref, thm-restate]{lipics-v2021}

\usepackage{xcolor}
\usepackage{paralist}
\usepackage{todonotes}

\usepackage[noend]{algorithm}
\usepackage[noend]{algpseudocode}

\algrenewcommand\algorithmicrequire{\textbf{Input:}}
\algrenewcommand\algorithmicensure{\textbf{Output:}}

\newcommand{\R}{\mathbb{R}}
\newcommand{\IND}{\mbox{\rm index}}
\newcommand{\bridge}{\mbox{\rm bridge}}
\newcommand{\Bridges}{\mbox{\rm Bridges}}
\newcommand{\UH}{\mbox{\rm UH}}
\newcommand{\Pgood}{P_{\mbox{\scriptsize\rm good}}}
\newcommand{\Pbad}{P_{\mbox{\scriptsize\rm bad}}}

\newtheorem{problem}{Problem}

\newcommand{\jeff}[1]{\textcolor{orange}{[Jeff: #1 ]}}

\title{Computing Planar Convex Hulls with a Promise} %TODO Please add

\author{Sepideh Aghamolaei}{Department of Computer Engineering, Amirkabir University of Technology, Tehran, Iran}{princess@aut.ac.ir}{https://orcid.org/0000-0003-1667-6323}{}

\author{Kevin Buchin}{Department of Computer Science, TU Dortmund, Germany \and \url{https://ae.cs.tu-dortmund.de/team/buchin/}}{kevin.buchin@tu-dortmund.de}{https://orcid.org/0000-0002-3022-7877}{Funded by the Deutsche Forschungsgemeinschaft (DFG, German Research Foundation) – Project number 550797858.}

\author{Timothy M. Chan}{Siebel School of Computing and Data Science, University of Illinois at Urbana-Champaign, USA}{tmc@illinois.edu}{https://orcid.org/0000-0002-8093-0675}{Partially supported by NSF Grant CCF-2224271.}

\author{Jacobus Conradi}{University of Copenhagen, Denmark}{jaco@di.ku.dk}{https://orcid.org/0000-0002-8259-1187}{Funded by the Carlsberg Foundation, grant CF24-1929.}

\author{Ivor Van der Hoog}{IT University of Copenhagen, Denmark}{ivva@itu.dk}{https://orcid.org/0009-0006-2624-0231}{Supported by the VILLUM Foundation grant (VIL37507) ``Efficient Recomputations for Changeful Problems''.}

\author{Vahideh Keikha}{Czech Academy of Science, Institute of Computer Science, Czechia}{keikha@cs.cas.cz}{https://orcid.org/0000-0003-2821-5903}{}

\author{Jeff M. Phillips}{Kahlert School of Computing, University of Utah, USA \and \url{https://users.cs.utah.edu/~jeffp/} }{jeffp@cs.utah.edu}{https://orcid.org/0000-0003-1169-2965}{thanks his support from NSF CCF-2115677, 2421782, and Simons Foundation
MPS-AI-00010515.}

\author{Benjamin Raichel}{Department of Computer Science,
  University of Texas at Dallas, USA \and
  \url{http://utdallas.edu/\string~benjamin.raichel} }
{benjamin.raichel@utdallas.edu}%
{https://orcid.org/0000-0001-6584-4843}%
{Work on this paper was partially supported by NSF AF Award
  CCF-2311179.}

\authorrunning{Aghamolaei, Buchin, Chan, Conradi, van der Hoog, Keikha, Phillips, and Raichel} %TODO mandatory. First: Use abbreviated first/middle names. Second (only in severe cases): Use first author plus 'et al.'

\Copyright{Aghamolaei, Buchin, Chan, Conradi, van der Hoog, Keikha, Phillips, Raichel} %TODO mandatory, please use full first names. LIPIcs license is "CC-BY";  http://creativecommons.org/licenses/by/3.0/

 %TODO mandatory: Please choose ACM 2012 classifications from https://dl.acm.org/ccs/ccs_flat.cfm 

 %TODO mandatory; please add comma-separated list of keywords

\category{} %optional, e.g. invited paper

\relatedversion{}

\acknowledgements{This paper is the result of research which began at the Banff International Research Station (BIRS) workshop titled ``Uncertainty in Combinatorial and Computational Geometry'' (26w5662).}%optional

\EventEditors{John Q. Open and Joan R. Access}
\EventNoEds{2}
\EventLongTitle{42nd Conference on Very Important Topics (CVIT 2016)}
\EventShortTitle{CVIT 2016}
\EventAcronym{CVIT}
\EventYear{2016}
\EventDate{December 24--27, 2016}
\EventLocation{Little Whinging, United Kingdom}
\EventLogo{}
\SeriesVolume{42}
\ArticleNo{XX}
\begin{document}
\nolinenumbers

\maketitle

\begin{abstract}

% alternative, shorter abstract
Computing the convex hull of a planar $n$-point set $P$ is one of the most fundamental problems in computational geometry. It has an $\Omega(n \log n)$ lower bound in the algebraic computation tree model of computation, and there are many convex hull algorithms that achieve this running time.
Classical results in computational geometry show that under special input assumptions it becomes possible to achieve sub-$O(n \log n)$ time algorithms to compute the convex hull. For instance, when the points are given as a sequence in lexicographical or angular order, the convex hull can be computed in linear time. Even under the weaker assumption that the sequence of points corresponds to the ordered vertices of a simple polygonal chain, linear-time algorithms exist. 

This naturally raises the question: Can the convex hull of a point set be computed in sub-$O(n \log n)$ time under weaker input assumptions? We answer this question positively. 
Under the promise that the input sequence of points contains the convex hull as a subsequence, we show a deterministic $O(n \sqrt{\log n})$-time algorithm to compute the convex hull of $P$.
Under randomisation, we achieve an expected running time of $O(n \log^{\varepsilon} n)$ for an arbitrarily small constant $\varepsilon > 0$.

We find our result surprising, as the points not on the convex hull  are allowed to behave adversarially to our convex hull construction algorithm. 
Yet, the promise that \emph{only} the points on the convex hull are sorted is sufficient for obtaining $o(n \log n)$-time algorithms. 
Finally, we show that the promise is not only sufficient but also tight: if we even slightly break this promise --- allowing even only one point on the convex hull to appear out of order --- then we show an adversarial $\Omega(n \log n)$-time lower bound.
As a consequence, we show that whether the given promise is true cannot be verified with fewer than $\Omega(n \log n)$ comparisons. Moreover, this implies that we negatively solve an open problem by L\"{o}ffler and Raichel, who conjectured sub-$O(n \log n)$ time algorithms for computing the convex hull of a supersequence that contains the convex hull as a subsequence.
\end{abstract}

\section{Introduction}

We assume as input a size-$n$ sequence $P$ of planar points.
The convex hull of $P$ is the smallest convex set that contains all points in $P$. It is represented by its boundary vertices in cyclical order.
Convex hulls are fundamental to computational geometry and are frequently and actively studied~\cite{brewer2024Dynamic, Goaoc2023Convex, Gaede2024Convex, Hoog2026ConvexALENEX, Wang2024Subpath, Wang2023Convex}.

There are many classical algorithms that compute the planar convex hull in $O(n \log n)$ time~\cite{andrew1979another,graham72,PreparataHong1977threedimensions}.
For points in three dimensions, Preparata and Hong~\cite{PreparataHong1977threedimensions} show an $O(n \log n)$-time algorithm. 
Chazelle~\cite{Chazelle1993optimalfixeddimension} presents an $O(n \log n + n^{\lfloor \frac{d}{2} \rfloor})$ time algorithm for when $P \subset \mathbb{R}^d$.
These algorithms are tight: the complexity of the convex hull of a $d$-dimensional set of $n$ points is in $\Theta( n^{\lfloor \frac{d}{2} \rfloor} )$~\cite{ziegler1994lectures}.
Furthermore, computing the convex hull has an $\Omega(n \log n)$ time lower bound in the algebraic computation tree model~\cite{Ben-Or83}.
%\end{comment}
\begin{comment}
These algorithms are tight: computing the convex hull has an $\Omega(n \log n)$ time lower bound in the algebraic computation tree model via a direct reduction from sorting~\cite{Ben-Or83}.
\end{comment}
 A variety of results also provide related lower bounds that exclude implicit constructions of the 3D convex hull~\cite{Erickson1995, Erickson1999New, EricksonSeidel1995}.

As discussed below, under suitable input assumptions sub-$O(n \log n)$-time algorithms exist for computing the convex hull in the plane. The aim of this paper is to provide such an algorithm under a considerably weaker assumption. We also show our assumption is tight, in the sense that just a slight violation of the assumption leads to an $\Omega(n \log n)$ lower bound.

\subparagraph{Linear-time algorithms.}
Graham~\cite{graham72} showed in 1972 a linear-time algorithm for when $P$ is sorted by their polar coordinates around $(0, 0)$.
%Polar coordinates are a geometrically natural concept, but they are also problematic from a computational point of view as exact angle computations are not supported by most models of computation (including the standard Real RAM as by Preparata and Shamos~\cite{preparata2012computational}). 
Andrew~\cite{andrew1979another} showed in 1979 a linear-time algorithm to construct the convex hull when $P$ is sorted by lexicographical order.
Note that if $P$ lies in general position, i.e., no two points share an $x$- or $y$-coordinate, then it suffices to have $P$ sorted by $x$-coordinate. 
If the sequence $P$ specifies a simple polygon then the convex hull can be computed in linear time. The first correct algorithm for this problem was given by McCallum and Avis~\cite{McCallumA79}, and many others followed, e.g.~\cite{Broennimann06space,GrahamY83,Lee83}. Melkman showed that the weaker assumption that $P$ specifies a simple polygonal chain is sufficient~\cite{Melkman87}.

Seidel~\cite{SEIDEL1985319} showed that $o(n \log n)$-time algorithms for sorted inputs are restricted to the plane only. He proved that even when given constantly many input sequences $I_P^1, I_P^2, \ldots I_P^k$, each containing $P \subset \mathbb{R}^3$ sorted along some direction, there remains an $\Omega(n \log n)$-time comparison-based lower bound for computing the three-dimensional convex hull of $P$.

\subparagraph{Tight sub-$O(n \log n)$ bounds.}
The $\Omega(n \log n)$ time lower bound is a \emph{worst-case lower bound}.
To circumvent this lower bound, several works offer a more refined analysis of convex hull construction that, under specific circumstances, achieves $o(n \log n)$ running times. 

Kirkpatrick, McQueen, and Seidel~\cite{kirkpatrick1986ultimate} consider output-sensitive analysis. If the convex hull of $P$ has $h$ vertices, then their algorithm runs in $O(n \log h)$ time and they show a matching comparison-based lower bound.
There is a long history of output-sensitive algorithms in $\mathbb{R}^3$~\cite{Chan1995Outputsensitive, Chan1995Output, Chand1970Convex, Clarkson1989, CHAZELLE199527, Edelsbrunner1991Threedim} and we consider the gift-wrapping algorithm by Chan~\cite{Chan1996Optimal} to be the ``textbook'' $\Theta(n \log h)$-time algorithm for point sets in $\mathbb{R}^2$ or $\mathbb{R}^3$.

Afshani, Barbay, and Chan~\cite{Afshani2017Instance} obtain even tighter bounds in the plane: 
for any fixed point set $P \subset \mathbb{R}^2$ and algorithm $\mathcal{A}$, define its \emph{input-sensitive running time} as its worst-case running time over all input sequences $I_P$  that store $P$. 
For a fixed point set $P$, the \emph{input-sensitive lower bound} is then the minimum over all algorithms $\mathcal{A}$ of their input-sensitive running time. 
Afshani, Barbay, and Chan~\cite{Afshani2017Instance} show that the input-sensitive lower bound for point sets $P$ is some entropy function and that the algorithm by Kirkpatrick, McQueen, and Seidel~\cite{kirkpatrick1986ultimate} asymptotically matches this lower bound for all point sets $P$. Van der Hoog, Rotenberg, and Rutschmann~\cite{Hoog2025Convex} show a simpler proof.  
Eppstein, Goodrich, Illickan, and To~\cite{EppsteinEtAlCCCG2025Entropy} provide an even tighter analysis by combining Andrew~\cite{andrew1979another} with the entropy function from~\cite{Afshani2017Instance} for which van der Hoog, Rotenberg, and Rutschmann~\cite{vanDerHoogRustenmann2025TightUniversal} provide matching lower bounds.

\subparagraph{Further related work: imprecise convex hulls.}
Another way of avoiding $\Omega(n \log n)$-time lower bounds is the \emph{preprocessing model of imprecise geometry} by  Held and Mitchell~\cite{held2008triangulating}.
In this model of computation, one first receives ``geometric advice'' about the planar input sequence $P$. 
This geometric advice is a labelled set of regions $\mathcal{R} = (R_1, \ldots, R_n)$ such that for all $i$, $p_i \in R_i$. 
After preprocessing the advice $\mathcal{R}$, we can achieve $o(n \log n)$-time convex hull algorithms.
If $\mathcal{R}$ are disjoint unit disks, then a variety of results~\cite{devillers2011delaunay,evans2011possible, ezra2013convex, held2008triangulating,loffler2025preprocessing,loffler2010delaunay,van2010preprocessing} can preprocess $\mathcal{R}$ in $O(n \log n)$ time, and compute the convex hull in (deterministic or randomised) $O(n)$ time. 
If $\mathcal{R}$ consists of $n$ distinct lines, then Ezra and Mulzer can preprocess these lines in $O(n \log n)$ time to construct the convex hull in  $O(n \cdot \alpha(n))$ expected time, where $\alpha(n)$ denotes the inverse-Ackermann function. 
Buchin, L\"{o}ffler, Morin, and Mulzer~\cite{buchin2011preprocessing} consider $\mathcal{R}$ to be unit disks that may overlap up to a \emph{ply} of $k$ and, after $O(n \log n)$-time preprocessing, achieve an $O(n \log k)$-time algorithm.
Recently,  de Berg, van der Hoog, Rotenberg, Rutschmann, and Wong show a variety of bounds for disks with bounded radii or bounded ply~\cite{deberg2025Instance-optimal}.

L\"offler and Raichel \cite{lr-pdch-26} consider the same setting, where the advice is given as a set of unit disks of ply at most $k$. 
They show that for any such advice $\mathcal{R}$, they can construct an ordered supersequence of points (i.e., an ordered sequence where each element is a point in $P$ that allows for repetitions) such that the vertices of the convex hull of $P$ form a subsequence of this supersequence. 
They reduce the problem of computing the convex hull under geometric advice to a natural and more abstract open problem:

\begin{problem}[Open Problem from \cite{lr-pdch-26}]
\label{Problem:supersequence}
Consider as input a sequence $I$ of size $m$, that stores all points of a size-$n$ point set $P$ (elements in $P$ may appear several times in the sequence). 
We additionally have the guarantee that there exists a subsequence $Q$ containing the vertices of the convex hull of $P$ sorted by $x$-coordinate. 
Given an input sequence where $m \in O(n)$, is it possible to compute the convex hull of $P$ in $o(n \log n)$ time? 
\end{problem}

\subparagraph{Contribution.}
In this paper, we observe that Problem~\ref{Problem:supersequence} has an even more natural formulation when $m = n$. 
In particular, we study the following problem: 
The input is a sequence $P$, under the promise that all points that lie on the (upper) convex hull appear in $P$ sorted by their $x$-coordinate. Our goal is to compute the (upper) convex hull in $o(n \log n)$ time.
Under this promise that $P$ contains the convex hull vertices as a subsequence, we show a deterministic $O(n \sqrt{\log n})$-time algorithm to compute the convex hull.
We additionally can achieve an expected running time of $O(n \log^{\varepsilon} n)$ for an arbitrarily small constant $\varepsilon > 0$. 

Our promise on the input is a direct generalisation of the promise by Graham~\cite{graham72} or Andrew~\cite{andrew1979another} who respectively
assume that the entire input is sorted by polar or $x$-coordinate. It is also strictly weaker than assuming that the input is ordered along a simple polygon or polygonal chain, since these assumptions imply that the ordered convex hull vertices are a subsequence.
One may wonder whether a weaker, and therefore more general, promise suffices to obtain $o(n \log n)$-time algorithms. 
We show that such generalisations are unlikely: 
Assume a weaker promise, where there exists a single unknown point $p^*$  and the sequence $P$ contains all its convex hull points (except for $p^*$) in sorted order.
We show that even for this slightly weakened promise, there exists an adversarial comparison-based $\Omega(n \log n)$ bound. 
We can additionally adapt our lower bound to show a negative answer for Problem~\ref{Problem:supersequence}: 
we show that given a supersequence of $P$ of size $m = 3n$, where the convex hull vertices of $P$ appear as a subsequence, there exists an $\Omega(n \log n)$-time lower bound. All our results also extend to the algorithmic problem of computing the Pareto front of $P$.

%L\"offler and Raichel \cite{lr-pdch-26} consider the problem of preprocessing a set of uncertain disks in the plane for computing the convex hull. Specifically, at query time a single point from each disk is given, i.e.\ the realisation of that disk, and the goal is to quickly report the convex hull of the resulting set of points. Different data structures can be built for the preprecessed disks, though \cite{lr-pdch-26} specifically consider the case when the data structure is simply an ordered sequence of the disks (with repetitions), with this promise that when the disks in the sequence are replaced with their realisations that the true convex hull appears as an ordered subsequence. Given such a sequence with the promise, the goal at query time is to quickly extract the convex hull from the sequence. \cite{lr-pdch-26} achieve this extraction using additional properties specific to given construction in the uncertain disk setting, however, this motivated \cite{lr-pdch-26} state the open problem of determining how quickly the convex hull can be extracted given only this promise. That is, the following problem can be consider completely independently of the original uncertain preprocessing motivation framework. 

\section{Preliminaries}

The input  is an ordered sequence of points $P =(p_1,\ldots,p_n) \subset \mathbb{R}^2$. For each point $p_i$ in the sequence, we denote its \emph{index} by $\IND(p_i)=i$, and its coordinates by $x(p_i)$ and $y(p_i)$. 
We assume $P$ to lie in general position, i.e., no two points share an $x$- or $y$-coordinate.

\subparagraph{Defining the convex hull promise.} We consider the problem of computing the \emph{convex hull} of the points $P$, denoted by $\mathrm{conv}(P)$.  Let $p_{\min}$ and $p_{\max}$ be the lexicographically smallest and largest points in $P$, respectively.
The \emph{upper convex hull} $\UH(P)$ is the part of the boundary of $\mathrm{conv}(P)$ running from $p_{\min}$ to $p_{\max}$ in clockwise order. It is represented by its sequence of vertices. The \emph{lower convex hull} is defined analogously, running from $p_{\max}$ back to $p_{\min}$. The complete boundary of $\mathrm{conv}(P)$ is obtained by concatenating the upper and lower convex hulls. Algorithms that construct the convex hull of $P$ typically focus on computing the upper convex hull. Computing the lower convex hull is done analogously, and they are then joined in linear time. We also take this approach in this paper. 
An input sequence is said to fulfil the \emph{convex hull promise} (refer to \Cref{fig:promise}), if for any two points $p,q \in P$ that lie on the convex hull we have that the index of $p$ is less than the index of $q$ if and only if the $x$-coordinate of $p$ is less also (formally,  $\IND(p)<\IND(q) \iff x(p)<x(q)$). 
Alternatively, it is more general to assume that the input fulfils both the \emph{upper} and \emph{lower convex hull promises} which are defined analogously for points on  the upper or lower convex hull, respectively.

Note that if $P$ fulfils the promise, it is not necessarily true that any (contiguous) subsequence of $P$ fulfils the promise --- which complicates designing divide-and-conquer strategies.
%Note that this is equivalent to requiring that the upper convex hull points, when ordered along the input sequence, form a simple polygonal curve. 

\begin{figure}[!b]
    \centering
    \includegraphics[width=\linewidth]{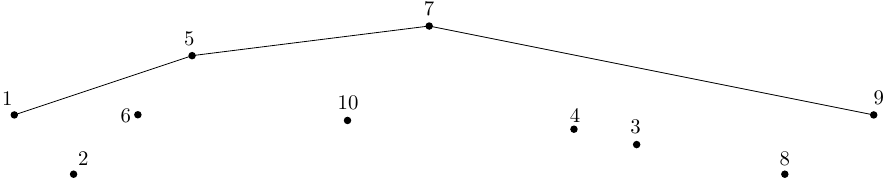}
    \caption{Illustration of the promise. The vertices defining the convex hull are $\{1,2,5,7,8,9\}$, which are sorted by their $x$-coordinate. The illustrated upper hull is defined by $\{1,5,7,9\}$.}
    \label{fig:promise}
\end{figure}

\subparagraph{Extra results: Pareto front.}
We note that we can also obtain results for the related problem of computing the \emph{Pareto front} of the points $P$.
The Pareto front is the sequence of points in $P$ that consists of exactly all points $p \in P$ for which there does not exist a point $q \in Q$ which has larger $x$- \emph{and} $y$-coordinate, ordered by $x$-coordinate. 
We can define a \emph{Pareto front promise} analogously: an input sequence is said to fulfil the \emph{Pareto-front promise} if for any two points $p,q$ on the Pareto front, the index of $p$ is less than the index of $q$ if and only if its $x$-coordinate is less also.
Our Pareto front results follow through similar techniques, and we thus henceforth focus on defining the appropriate convex hull terminology.

\subparagraph{Defining $L$-bad pairs.}
In this paper, we show that we can circumvent the $\Omega(n \log n)$-lower bound for computing these structures for point sets that adhere to the convex hull (or Pareto front) promise. 
Our main observation to enable faster-than-$O(n \log n)$ computations is the following: Suppose we find two points $p$ and $q$ where their index-ordering does not match their $x$-ordering, then we know that at most one of $p$ and $q$ can be on the convex hull. To obtain $o(n\log n)$-time algorithms, we cannot spend time sorting the input to detect these pairs. Instead, we design an algorithm that detects these out-of-order pairs approximately using what we call $L$-bad pairs:

\begin{definition}
Let $L$ be a set of vertical lines in $\R^2$.
Two points $p$ and $q$ are said to be an \emph{$L$-bad pair}, if $p$ and $q$ are separated by a line in $L$ and 
either $\IND(p)<\IND(q)$ and $x(p)>x(q)$, or,  $\IND(p)>\IND(q)$ and $x(p)<x(q)$.
\end{definition}

\subparagraph{Slabs and bridges.}
Our approach relies on computing a suitable set of vertical lines $L$ and finding sufficiently many $L$-bad pairs. 
To this end, we observe that such a set $L$, where no two lines share an $x$-coordinate, partitions $\R^2$ into \emph{vertical slabs}, which we denote by $\Gamma(L)$. 
%For some suitable set of lines $L$, we subsequently compute \emph{convex hull bridges}.

%For any point set $P$, the \emph{convex hull} is the sequence of points along the smallest convex area containing $P$. 
Interpreting the edges of the upper convex hull as half-open segments, the upper convex hull has the property that any vertical line $\ell$ intersects either no edge of $\UH(P)$, or exactly one edge of $\UH(P)$. For any vertical line $\ell$, we denote by $\bridge(P,\ell)$ the unique edge of $\UH(P)$ that is intersected by $\ell$. For a set $L$ of vertical lines we define the set of bridges $\Bridges(P,L)=\bigcup_{\ell\in L}\{\bridge(P,\ell)\}$. 

\begin{lemma}\label{lem:abstract-bridge-compute}\cite[Theorem 3.1]{kirkpatrick1986ultimate} 
    Given $P$ and any vertical line $\ell$, we can compute $\bridge(P,\ell)$ in $O(|P|)$ time.
\end{lemma}
\begin{comment}
\begin{proof}
    In linear time, we may transform $P$ such that $\ell$ is the $y$-axis. Consider the linear program with variables $m,b$:
    \[\min_{m,b}\ b\qquad\text{subject to}\qquad y(p)\le m\cdot x(p)+b \quad \text{for all } p\in P.\]
    Any feasible solution $(m,b)$ defines a line $L_{m,b}:\ y=mx+b$ that lies above every point of $P$. Since $\ell$ is the $y$-axis, the value of $b$ is exactly the $y$-coordinate where $L_{m,b}$ intersects $\ell$. Thus, minimising $b$ is equivalent to finding, among all lines that lie above $P$, the one whose intersection with $\ell$ is as low as possible.
    This is exactly the supporting line of the upper hull of $P$ that crosses $\ell$. Therefore, the optimal line is the line through the $\bridge(P,\ell)$ as if there were a lower such line crossing $\ell$, then the segment defining the bridge would not be the upper tangent crossing $\ell$.
    The program is a two-dimensional linear program with $O(|P|)$ constraints, so it can be solved in $O(|P|)$ time \cite{Megiddo1984}.  At an optimum, the line $L_{m,b}$ is tight for at least one point of $P$, and for the supporting line corresponding to the bridge it is tight for the two endpoints of the bridge. Hence, after computing $(m,b)$, a linear scan through $P$ can identify all points satisfying $y(p)=mx(p)+b$.
    Among these tight points, the two extreme ones on the line are exactly the two endpoints defining $\bridge(P,\ell)$.
\end{proof}
\end{comment}

\section{Warm up: A deterministic $O(n\sqrt{\log n})$ algorithm}

As is convention, we focus on computing only the upper hull $\UH(P)$. The algorithm for computing the lower hull is symmetric and merging them yields the convex hull.

We start with defining a \emph{slab decomposition} on a set $L$ of vertical lines $\{\ell_1, \ell_2, \ldots, \ell_{|L|}\}$.  For each $\ell_i$ we associate it with its $x$-value in $\R$, and for notational convenience we also define $\ell_0 = -\infty$ and $\ell_{|L|+1} = \infty$.  
Each vertical slab $\gamma_i \in \Gamma(L)$ is defined by a left $\ell_i$ and right $\ell_{i+1}$ line in $L$ marking its end points.  Let $P_{\gamma_i}$ be the subset of $P$ that satisfies $\ell_i \leq x(p) < \ell_{i+1}$.  
%We augment the slab decomposition to also store the bridges $\Bridges(P,L)$.
For an $\ell_i\in L$, we denote by $v_i^-,v_i^+ \in P$ the left and right end points of $\bridge(P,\ell_i)$. Now define the safe points $P^{(\gamma_i)} \subset P_{\gamma_i}$ as all points $p$ in $P_{\gamma_i}$ that meet all of the following requirements:
\begin{itemize}
    \item $p$ does not lie under one of the bridges $\bridge(P,\ell_i)$ and $\bridge(P,\ell_{i+1})$;
    \item the index of $p$ is at least that of $v_i^+$;
    \item the index of $p$ is at most that of $v_{i+1}^-$.
\end{itemize}
%points that do not violate the points defining these bridges in terms of $x$ order, and do not violate the points defining these bridges in terms of index.  
Formally, $P^{(\gamma_i)} = \{ p \in P \mid  \IND(v_i^+)\leq\IND(p)\leq\IND(v_{i+1}^-)$ \& $x(v_i^+)\leq x(p)\leq x(v_{i+1}^-)\}$. 
%\ben{+/- in algorithm 1 could be confused with notation here. Also, may be worth remarking $x(v_i^+)>x(v_{i+1}^-)$ is possible if edge goes across slab, definition still fine though as empty set then.}  
We define the slab decomposition on $L$ as the splitting of $P$ into the sets $P^{(\gamma_i)}$ and $P_{\gamma_i} \setminus P^{(\gamma_i)}$.

%For every $\gamma\in\Gamma(L)$ let $P_\gamma$ be the subset of $P$ that is restricted to $\gamma$. Then $\Bridges(\bigcup_{\gamma}P_\gamma,L)$ can be computed in $O(|L|\cdot\max_\gamma|P_\gamma|)$ time.
%\todo{This makes no sense. Per definition, $P = \bigcup_{\gamma}P_\gamma$ and so you are just computing $\Bridges(P, L)$. If that is what you mean, then don't define all the $P_\gamma$'s yet and use them in the proof interior.}\jacobus{$P_\gamma$ appears in the running time. Also the $P_\gamma$'s need to be given as input. So I cannot move it to the interior. }
\begin{comment}
\begin{proof}
    Graham's scan on the sets $P_\gamma$ requires $O(|L|)$ tangent-finding operations. By \Cref{lem:abstract-bridge-compute} each such operation takes $O(\max_\gamma|P_\gamma|)$ time.\todo{I know exactly what you mean, and still want significantly more details ;)}
\end{proof}
\end{comment}

\begin{lemma}\label{lem:slabs-correct}
 Points in $P_{\gamma_i} \setminus P^{(\gamma_i)}$ cannot be on the upper convex hull if the promise holds.  
\end{lemma}
\begin{proof}%\ben{Issues with some sentences in this proof.}
The points  $v_i^+$ and $v_{i+1}^-$ associated with slab $\gamma_i$ are by construction endpoints of edges of the upper hull of $P$.  By definition of the upper convex hull, any point $p \in P_{\gamma_i}$ which has $x(p)$ in the slab, can only be on the upper hull if its $x(p)$ lies between $x(v_i^+)$ and $x(v_{i+1}^-)$.  Moreover, if the promise holds, then $p$ must also have index value between that of $v_i^+$ and $v_{i+1}^-$.  Points in $P_{\gamma_i} \setminus P^{(\gamma_i)}$ violate, by definition, one of these two conditions.
\end{proof}

We can therefore construct the upper convex hull by a two-step process: (1) build the slab decomposition with some smartly computed set $L$, and then (2) recurse on each safe subset $P^{(\gamma)}$. This is how our deterministic algorithm (Algorithm \ref{alg:prob1}) will proceed.  The question remains how to choose the slabs, efficiently construct the decomposition, and analyse the result.  We will require a \emph{balanced slab decomposition} with $|L| = b$ defined such that each $P_\gamma$ for $\gamma \in \Gamma(L)$ has size $O(n/b)$.  

%We begin by analysing the construction of the slab decomposition. 
%We do this in two steps.   
To compute such a decomposition, we first compute $b$ lines $L$ (for a parameter $b$ to be fixed later), and for every $\gamma\in\Gamma(L)$ compute the set of points $P_\gamma$. Its construction time will depend on an unknown number $z$ of $L$-bad pairs. Luckily, the bigger $z$, the smaller the safe subsets $P^{(\gamma)}$ of $P_\gamma$.
Second, we compute bridges and filter out points in $P_{\gamma_i} \setminus P^{(\gamma_i)}$.

\begin{lemma}\label{detLem1}
    Let $P$ be a sequence of $n$ points in $\R^2$ (we do not assume that $P$ fulfils the promise). Given a parameter $b$, we can generate a set $L$ of $O(b)$ vertical lines, and compute $P\cap\gamma$ for all slabs $\gamma\in\Gamma(L)$ in $O(n+z\log b)$ time for some $z$, such that
    \begin{itemize}
        \item $|P\cap\gamma|\leq O(n/b)$ for each $\gamma\in\Gamma(L)$, and
        \item there exist at least $z$ disjoint $L$-bad pairs.
    \end{itemize}
\end{lemma}
\begin{proof}
    We incrementally add points in index order and maintain a set $L'$ of lines. Initially $L'=\emptyset$. Each point that has been seen is marked \emph{live} or \emph{dead}. We maintain the invariant that each slab $\Gamma(L')$ contains at most $n/b$ live points, except the rightmost slab, which contains strictly between $0$ and $2n/b$ live points, unless there are currently $0$ live points. We maintain the list of all live points inside each slab of $\Gamma(L')$.

    To insert the next unseen point $p$, we mark $p$ as live and do the following:

    \begin{itemize}
        \item Case 1: $p$ is in the rightmost slab $\gamma$ of $\Gamma(L')$. If $\gamma$ now has $2n/b$ live points, we compute the median $x$-coordinate of these $2n/b$ points. We split $\gamma$ into its left and right subslab according to the median line $\ell$, and add $\ell$ to $L$, restoring the invariant. This takes $O(n/b)$ time but is done a total of $O(b)$ times, since Case 1 needs to be triggered $n/b$ times before it reaches the $2n/b$ limit after a split.  
        %We add this median line $\ell$ to $L$.
        \item Case 2: $p$ is not in the rightmost slab $\gamma$ of $\Gamma(L')$. Let $q$ be any live point in $\gamma$. Then $(p,q)$ is a $L'$-bad pair and we change the marks of $p$ and $q$ to dead. If $\gamma$ now has $0$ live points, we delete the rightmost line from $L'$, repeatedly until the rightmost slab has at least $1$ live point, restoring the invariant.
    \end{itemize}

    \noindent
    After this procedure, let $L$ be the set of all the lines that have ever been created ($L$ contains $L'$, but also includes lines that have been deleted from $L'$). Then $|L|,|L'| \in O(b)$. Moreover, we observe that all the pairs of dead points from Case 2 are $L$-bad. 

    We can maintain for each $\ell \in L \setminus L'$ in which slab $\beta \in \Gamma(L')$ it lies in by keeping $L$ in sorted order.  At no asymptotic overhead, we can also maintain for each slab of $\Gamma(L')$ the list of all live points in the slab. This allows us to map all live points to its slab of $\Gamma(L)$. In particular, we achieve this by iterating over each slab $\beta\in\Gamma(L')$ and comparing each of its $O(n/b)$ live points with each line of $L$ which is in $\beta$. If we denote by $n_\beta$ be the number of lines of $L$ in $\beta$, this takes $O(\sum_\beta(n/b)n_\beta)=O((n/b)|L|) = O(n)$ time. 

    Now finally we deal with the $L$-bad pairs.  In $O(\log b)$ time per point that is part of a detected $L$-bad pair, we locate their corresponding slab. Whenever such an insertion results in a slab $\gamma$ containing more than $n/b$ points, we again compute a median in $O(n/b)$ time, and distribute the points according to this added median line, and add this median line to $L$.  While this action takes $O(b)$ time, the total time across all $2z = O(n)$ bad points can be bounded by $O(n)$ time.  The first event in each slab might cost $O(b)$, but the second median splitting event can be amortised over $(n/b)/2$ insertions; so the total time is $O(b (n/b)) = O(n)$ for first instances in each of $O(b)$ slabs, and $O((n/b)/(n/b)) = O(1)$ amortised time for each subsequent insertion in a slab.     
    This ensures that afterwards $|P\cap \gamma|\leq n/b$ for every $\gamma\in\Gamma(L)$. Overall, this takes $O(n + z\log b)$ time. Afterwards we know $P\cap\gamma$ for every $\gamma\in\Gamma(L)$.
    %We still need to prove that $|P \cap \gamma| = O(n/b)$.  This was true among the live points by the incremental construction, but we need to account for the L-bad points inserted which may not violate the ultimate x-ordering vs. index constraint.  
\end{proof}

\begin{figure}
    \centering
    \includegraphics[width=\linewidth]{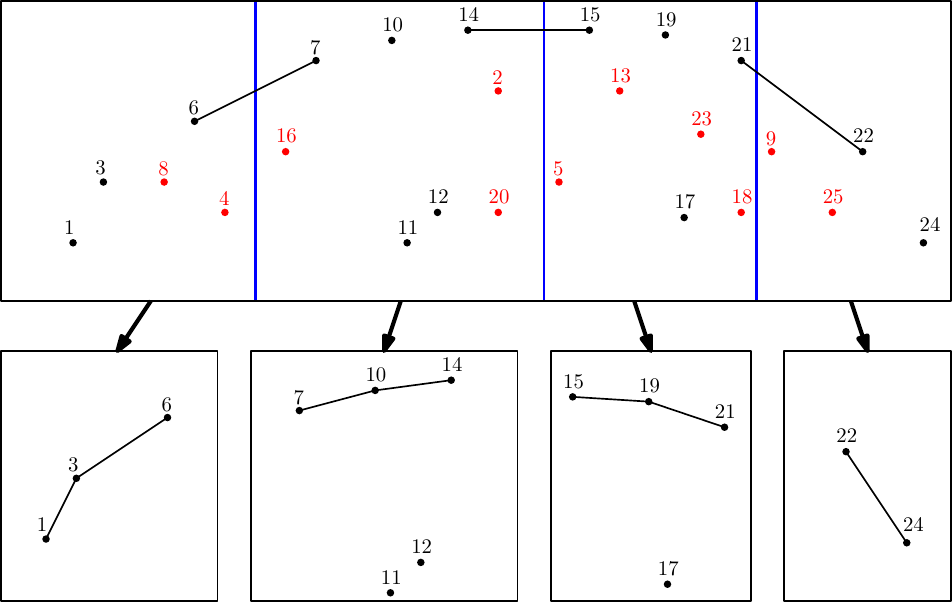}
    \caption{Illustration of recursive call tree of \Cref{alg:prob1}. Red points are pruned, and are not passed along in the recursion step, as either their index does not lie between the indices of $v^-$ and $v^+$, or, their $x$-coordinate does not lie between the $x$-coordinates of $v^-$ and $v^+$.}
    \label{fig:prunetree}
\end{figure}

%We can now compute the slab decomposition in $O(nb + z \log b)$ time by calling Lemma \ref{lem:abstract-bridge-compute} on all of $P$ for each $\ell \in L$.  But we can improve the first term to only $O(n)$ using a Graham's scan-style approach, and only invoking Lemma \ref{lem:abstract-bridge-compute} using a local subset of $P$.

Based on the decomposition of \Cref{detLem1} we can now compute a slab decomposition. To compute the bridges we use a Graham's scan-style approach together with \Cref{lem:abstract-bridge-compute}.

\begin{lemma}\label{cor:compute-slab-decomposition}
    Given a sequence $P$ of $n$ points in $\R^2$, and a parameter $b$, we can compute a set of $O(b)$ vertical lines $L$ and its induced slab decomposition of $P$.  For some $z$, where there are at least $z$ $L$-bad pairs in $P$ this takes $O(n+z\log b)$ time.
\end{lemma}
\begin{proof}
    First, we invoke \Cref{detLem1} to obtain a sorted sequence of vertical lines $L=\{\ell_1,\ldots,\ell_{|L|}\}$ 
    (recall $\ell_0 =-\infty$ and $\ell_{|L|+1} = \infty$ for notational convenience),
    and $P_i=P\cap\gamma_i$ for every slab $\gamma_i\in\Gamma(L)$ defined by $\ell_i$ and $\ell_{i+1}$. We now proceed in a Graham's scan-style fashion.

    While incrementing an integer $I$ from $1$  to $|L|-1$, we maintain $\Bridges(\bigcup_{i\leq I}P_i,\{\ell_1,\dots,\ell_{I}\})$. These are the edges of the convex hull of $\bigcup_{i\leq I}P_i$ that cross $\{\ell_1,\dots,\ell_I\}$. Every such edge starts in some set $P_i$ and ends in some $P_j$ with $j>i$. Concretely, we maintain the subsequence of $1,\ldots,I$ of the sets defining the edges of $\Bridges(\bigcup_{i\leq I}P_i,\{\ell_1,\dots,\ell_{I}\})$, similar to Graham's scan. Let $S_I$ be this sequence. In particular, if for some $k\leq I$ we have that $k\not\in S_I$, 
    then $k$ is also in no other $S_{I'}$ with $I'>I$, as no point of $P_k$ can be on the upper hull. 
    Also, as $P_{I+1}$ is right of $\ell_{I+1}$, and all other sets that have been considered so far lie left of $\ell_{I+1}$, $S_{I+1}$ consists of a prefix from $S_I$ together with $I+1$.

    To mimic Graham's scan, we consider adding $I+1$ to the subsequence, adding $P_{I+1}$ to the maintained set. Let $k$ be the current rightmost element in $S_I$, and let $j$ be the predecessor of $k$ in that subsequence. Since $\ell_{I+1}$ separates $P_k$ from $P_{I+1}$, \Cref{lem:abstract-bridge-compute} allows us to compute the bridge of $P_k\cup P_{I+1}$ across $\ell_{I+1}$ in $O(\max_i |P_i|)$ time. If this bridge is compatible with the currently maintained bridge incident to $P_k$, then $k$ is also in $S_{I+1}$ and we are done. Otherwise, $k$ is not in $S_{I+1}$, so we remove it from the maintained subsequence and continue with its predecessor. Thus, exactly as in Graham's scan, we repeatedly delete the rightmost index until we find the unique predecessor $P_j$ that is still connected to the new rightmost set $P_{I+1}$ by a bridge. This one bridge then replaces all deleted bridges to its right. In the end, we add $I+1$ to the subsequence resulting in $S_{I+1}$.

    As every set $P_i$ is inserted once and deleted at most once, the total number of bridge computations is $O(|L|)$. Since each such computation takes $O(\max_i |P_i|)$ time by \Cref{lem:abstract-bridge-compute}, the overall running time (outside the
    application of \Cref{detLem1}) is $O(|L|\cdot \max_i |P_i|) = O(b\cdot (n/b)) = O(n)$. Afterwards, we have computed $P_\gamma=P\cap \gamma$ for every $\gamma\in\Gamma(L)$ as well as the set $\Bridges(P,L)$. The last thing to compute for the slab decomposition is the subset of safe points $P^{(\gamma)}$. This can be done with a linear scan of $P_{\gamma}$, comparing the indices and $x$-coordinates of the points in $P_\gamma$ to the indices and $x$-coordinates of the endpoints of the (at most) two bridges of the two lines defining $\gamma$. Thus, after a total of $O(n + z\log b)$ time we have computed a slab decomposition of $P$.
\end{proof}

\begin{algorithm}
\caption{A deterministic $O(n\sqrt{\log n})$-time convex hull algorithm with a promise}\label{alg:prob1}
\begin{algorithmic}[1]
\Require Sequence $P$ of $n$ points in $\R^2$ fulfilling the convex-hull promise, fixed parameter $b$
\Ensure $\UH(P)$
\State If $n\leq b$, compute and return $\UH(P)$ via Graham's scan.
%\State Apply \Cref{detLem1} to obtain a set of $O(b)$ lines $L$, and $P\cap\gamma$ for all slabs $\gamma\in\Gamma(L)$.\label{line:1-2}
%\State Check that $|P\cap\gamma|=O(n/b)$ for all $\gamma\in\Gamma(L)$, else compute $\UH(P)$ via Graham's scan.
%\State Compute $\Bridges(P,L)$ via \Cref{cor:asbstract-bridges-compute}.\label{line:1-4}
\State Apply \Cref{cor:compute-slab-decomposition} to obtain a set of $O(b)$ lines $L$, and its resulting slab decomposition. 
\For{$\gamma \in \Gamma(L)$}  % with left side $\ell_i \in L$ and right side $\ell_{i+1} \in L$}
%\State Let $v_i^+$ be the right endpoint of $\bridge(P,\ell_i)$.
%\State Let $v_i^-$ be the left endpoint of $\bridge(P,\ell_{i+1})$.
%\State Let $P^{(\gamma_i)}=\{p\in P\cap\gamma_i \mid \text{$\IND(v_i^+) \leq \IND(p) \leq \IND(v_{i+1}^-)$ \& $x(v_i^+)\leq x(p)\leq x(v_{i+1}^-)$}\}$.
\State Recursively compute $\UH(P^{(\gamma)})$.
\EndFor
\addtocounter{linenumber}{-1}
\State Concatenate $\{\UH(P^{(\gamma)})\}$ and $\Bridges(P,L)$ and return $\UH(P)$.
\end{algorithmic}
\end{algorithm}

\begin{theorem}\label{thm:det}
    \Cref{alg:prob1} computes the upper convex hull of a sequence $P$ of $n$ points, assuming that the input sequence fulfils the convex-hull promise, in time $O(n \sqrt{\log n})$.
\end{theorem}
\begin{proof}
\Cref{lem:slabs-correct} implies that the upper convex hull $\UH(P)$ does not contain any points of $P_\gamma\setminus P^{(\gamma)}$. Conversely, if $L$ is any set of vertical lines then the upper hull is defined by $\Bridges(P,L)$, and all $\UH(P^{(\gamma)})$ for $\gamma\in\Gamma(L)$, which is exactly what the algorithm computes, thus the upper hull is computed correctly.

For the running time, applying \Cref{cor:compute-slab-decomposition} takes $O(n + z\log b)$ time, where there are at least $z$ $L$-bad pairs in $P$. %, and applying \Cref{cor:asbstract-bridges-compute} to every slab takes $O(n)$ time, as there are $b$ slabs, and each slab contains $O(n/b)$ points.
In the \textbf{for}-loop, observe that for each $L$-bad pair $(p,q)$,
one of its points will be pruned:
if $p$ is left of $\ell$ and $q$ is right of $\ell$
for some line $\ell\in L$, and let $v$ be the right endpoint of $\bridge(P,\ell)$;
if $\IND(q)< \IND(v)$, then $q$ will be pruned;
otherwise, $\IND(p) > \IND(q) \ge \IND(v)$, so $p$ will be pruned.
Thus, $\sum_\gamma |P^{(\gamma)}|\le n-z$.

Overall, the run time satisfies the following recurrence for $n > b$:
\[ T(n)\ \le\ \max_{\scriptsize\begin{array}{c}z,n_1,n_2,\ldots:\\\sum_i n_i\le n-z\\ \max_i n_i \le O(n/b)\end{array}}
  \left(O(n+z\log b) +  \sum_i T(n_i) \right).
\]
Now consider $F(n)=C\left(n\frac{\log n}{\log b} + n\log b\right)$ for sufficiently large $C$. We show that $T(n)\leq F(n)$ via induction.
For the base case $n \le b$, we na\"ively have $T(n)\le O(n\log n)\le O(n\log b)\leq F(n)$. 
For $n\geq b$ we inductively assume that $T(n_i) \leq F(n_i)$, and aim to show $T(n) \leq F(n)$.  
Let $S=\sum_in_i\leq n-z$. Since $\max_in_i\leq \alpha n/b$ for some absolute constant $\alpha$, we obtain
\[\sum_in_i\log n_i\leq S\log(\alpha n/b)=S\log n - S\log  b + O(S).\]
Hence
\[\sum_iF(n_i)\leq C\left(\frac{S\log n}{\log b}-S +O\left(\frac{S}{\log b}\right) +  S\log b \right) \leq C\left(\frac{S\log n}{\log b}+ S\log b - \Omega(S)\right).\]
This together with $S\leq n-z$ and our inductive assumption that $T(n_i) < F(n_i)$ implies
\[
\sum_i T(n_i) \leq
\sum_iF(n_i) \leq 
C\left(\frac{S\log n}{\log b}+ S\log b - \Omega(S)\right)=F(n)-Cz\log b-\Omega(CS).
\]
Overall, since $n\leq S+z$, we obtain
\[T(n)\leq O(n+z\log b) + \sum_iT(n_i)\leq O(n+z\log b) + F(n)-Cz\log b - \Omega(CS)\leq F(n).\]

Choosing $b$ with $\log b=\sqrt{\log n}$ yields $T(n)\le O(n \sqrt{\log n})$.
\end{proof}

\section{A randomised $O(n\log^{\varepsilon}n)$ algorithm}

Note that Algorithm~\ref{alg:prob1} would be linear if we did not have to spend $O(\log b)$ time per $L$-bad pair. In the following we attempt to reduce this logarithmic factor by introducing more recursion and using randomisation. 
We replace the application of \Cref{cor:compute-slab-decomposition} with a recursive subroutine to compute $\Bridges(P,L_0)$ of a random sample $L_0$.
%of size $b\log b$ in (roughly) $2^{O(\sqrt{\log \log b})}n\log\log n$ time. 
%To achieve this, we cannot explicitly locate for all the bad points in which slab they lie.
We cannot explicitly locate for all points in which slab they lie.
Instead, we compute a subset $L$ of $L_0$ as slabs on the fly.
We split the point set into good points, for which we know in which slab they live, and a set of bad points. We preprocess good points  by grouping them into small groups and computing the convex hull of each group. We then recurse on $L$. That is, we compute $\Bridges(P,L)$, with the knowledge that $P$ consists of a set of yet-to-be-located points, and a set of well-behaved small convex hulls that are each constrained to a slab in $\Gamma(L)$. After this first recursive call, we have correctly computed $\Bridges(P,L)$. From these, we can finally filter $P$ into its slabs by their index and $x$-coordinate in linear time, and then recurse in each slab $\Gamma(L)$, to compute $\Bridges(P,L_0\setminus L)$. To this end, we define the following technical subproblem:
%This is possible, since we no longer explicitly locate the slab, a bad pair goes to, and instead---on the fly---refine slabs instead.

\begin{problem}[line-guided bridge computation]
\label{prob:lineguided}
Fix a number $g$.
The input consists of an $x$-sorted set $L_0$ of at most $\lambda$ vertical lines, and a point sequence $P_0$ fulfilling the convex-hull promise, where
$P_0$ is additionally partitioned into two subsets $P$ and $Q$:
\begin{itemize}
\item $P$ is a subsequence of $P_0$, and
\item $Q$ consists of $\mu$ \emph{groups}, each of
$\le g$ points, where each group is entirely inside a slab of $\Gamma(L_0)$,
and for each group we are given the slab it is in as well as the group's upper hull.
\end{itemize}
The task is to compute $\Bridges(P_0,L_0)$.
\end{problem}

%This subproblem with a set $L_0$ of lines, and $Q=\emptyset$ simply computes $\Bridges(P_0 
%,L_0)$.

We now strengthen \Cref{detLem1} to start with a set $L_0$, from which we select a well-spread set of lines with very similar guarantees to \Cref{detLem1}, except we do not locate for our identified bad points in which slab they lie. 

\begin{lemma}\label{lem2}
Let $P$ be a sequence of $n$ points in $\R^2$ (we do not assume that it fulfils the promise). Given an $x$-sorted list $L_0$ of $\le\lambda$ vertical lines, and a parameter $b\le O(\lambda/\log\lambda)$,
we can partition $P$ into two subsets $\Pgood$ and $\Pbad$, generate a set $L\subseteq L_0$ of $O(b)$ lines, and
compute $\Pgood\cap \gamma$ for all slabs $\gamma\in\Gamma(L)$, in $O(n+\lambda)$ time, such that
\begin{itemize}
\item there are $\le\lambda/b$ lines of $L_0$ inside each slab $\gamma\in\Gamma(L)$, and
\item $\Pbad$ consists of $|\Pbad|/2$ disjoint $L$-bad pairs.
\end{itemize}
\end{lemma}
\begin{proof}
We incrementally iterate over all points in $P$ in index order and maintain a set $L'\subseteq  L_0$ of lines.
Initially, $L'=\emptyset$.

Each point that has been seen is marked \emph{live}, \emph{dead}, or \emph{located}.
We maintain the invariant that
each slab of $\Gamma(L')$ contains at most $n/b$ live points, except the rightmost slab,
which contains strictly between 0 and $2n/b$ live points, unless there are $0$ live points. As our algorithm splits slabs that contain too many live points, it may be impossible to maintain this invariant. This occurs in particular, when there is a too-full slab that is defined by two consecutive lines in $L_0$,  which therefore cannot be split by adding to $L'$ more lines from $L_0$. We then mark points in this slab as located instead. 

As we iterate over $P$, we maintain $L'$ and also the list of all live points inside each slab of $\Gamma(L')$. When we iterate over the next unseen point $p$, we mark $p$ as live and do the following:
\begin{itemize}
\item Case 1: $p$ is in the rightmost slab $\gamma$ of $\Gamma(L')$.
If, after adding $p$, $\gamma$ now has $2n/b$ live points,
we compute the median $x$-coordinate of these $2n/b$ points 
and find the slab $\gamma_m\in\Gamma(L_0)$ containing this median.
Say $\gamma_m$ has as its boundary $\ell_m^-$ and $\ell_m^+$.
We change the marks of all live points inside $\gamma_m$ to \emph{located}.
We add $\ell_m^-$ to $L'$: this splits $\gamma$ into at most three sub-slabs, where at most two contain any live points (at most $n/b$ live points each), restoring the invariant.
This takes $O(n/b + \log\lambda)$ time, and is done a total of $O(b)$ times.

\item Case 2: $p$ is not in the rightmost slab $\gamma$ of $\Gamma(L')$.
Let $q$ be any live point in $\gamma$.
Then $(p,q)$ is an $L'$-bad pair, and we change the marks of $p$ and $q$ to dead.
If $\gamma$ now has 0 live points, we delete the rightmost line from $L'$,
repeatedly until the rightmost slab has at least 1 live point, restoring the invariant.
\end{itemize}

At the end, 
let $L''$ be all the lines that have ever been created ($L''$ contains $L'$ but also includes lines that have been deleted from $L'$).  Then $|L'|,|L''| \in O(b)$. 
Also let $L'''$ be the $b$ quantiles of $L_0$.
Define $L=L''\cup L'''$.  Sort $L$ in $O(b\log b)$ time.
All the pairs of dead points discovered are $L''$-bad and thus $L$-bad.
Let $\Pgood$ be all the live points and all the located points, and $\Pbad$ be all the dead points.

We know the list of all live points inside each slab of $\Gamma(L')$.
We can also generate the list of all live points inside each slab of $\Gamma(L)$. 
We know for each $\ell \in L \setminus L'$ which slab $\beta \in \Gamma(L')$ it lies in.  
For a fixed slab $\beta \in \Gamma(L')$, we consider all $\ell \in  L \setminus L'$ that lie in $\beta$ and we compare the $O(n/b)$ live points in $\beta$ to these lines to identify which slab of $\Gamma(L)$ they lie in. 
Letting $n_\beta$ be the number of lines of $L \setminus L'$ in $\beta$, this takes
time $O(\sum_\beta (n/b)n_\beta)=O((n/b)|L|) = O(n)$.

On the other hand, the located points lie in $O(b)$ slabs of $\Gamma(L_0)$.  For each located point, we know per construction  the slab of $\Gamma(L_0)$ containing it. In $O(b\log b)$ time, we can pre-determine the
slab of $\Gamma(L)$ containing each of these $O(b)$ possible slabs of $\Gamma(L_0)$.
Thus, we know the slab of $\Gamma(L)$ containing each located point. 

It follows that we know $\Pgood\cap\gamma$ for every $\gamma\in\Gamma(L)$.
\end{proof}

With this algorithmic tool that coarsens $L_0$ to $L$, but also quickly splits the points into located and bad points, we can now solve the line-guided bridge computation problem  (Problem~\ref{prob:lineguided}) in $n\cdot2^{O(\sqrt{\log\log\lambda})}$ time via a recursive method that is outlined in Algorithm \ref{alg:prob2}.  

\begin{figure}
    \centering
    \includegraphics[width=\linewidth]{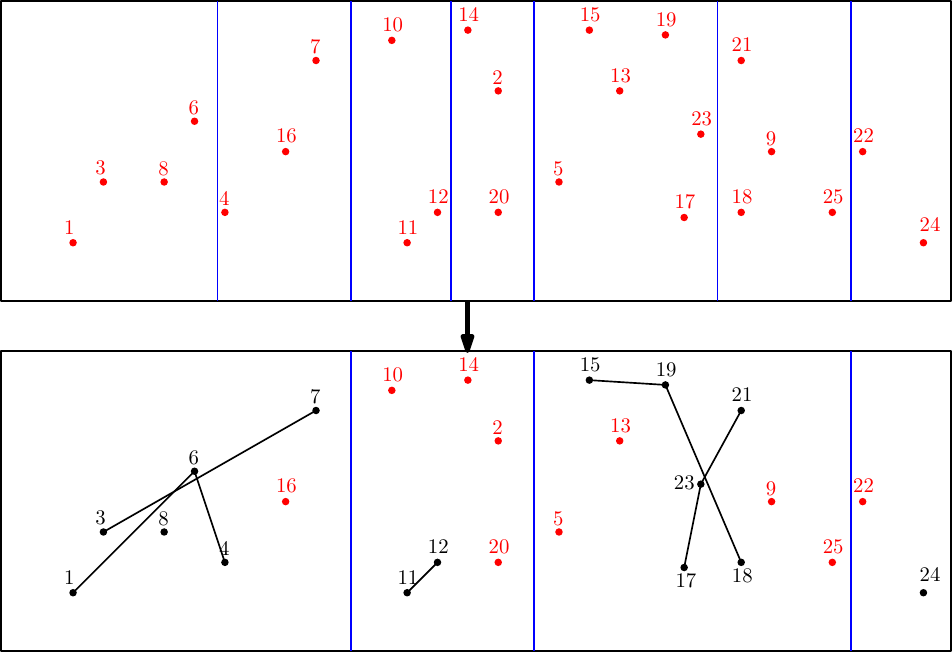}
    \caption{The recursive call from Line 4 of \Cref{alg:prob2}. We obtain a subset of lines, the good (black) points are split into groups and we compute their upper hulls.  For (red) points we do \emph{not} know in which slab they live. From the recursive call we get the bridges of the three blue lines.}
    \label{fig:lambdaTree}
\end{figure}

\begin{algorithm}
\caption{Line-guided bridge computation}\label{alg:prob2}
\begin{algorithmic}[1]
\Require Set $L_0$ of $\lambda$ lines, sequence $P$ of $n$ points, $\mu$ groups $\{Q_i\}$, and parameters $b$ and $g$
\Ensure $\Bridges(P\cup\bigcup_iQ_i,L_0)$
\State If $\lambda\leq O(1)$ compute $\Bridges(P\cup\bigcup_iQ_i,L_0)$ via 2D LPs over $n+\mu$ convex $g$-gons.
\State Via \Cref{lem2}, obtain $\Pgood,\Pbad\subseteq P$, $L\subseteq L_0$, and
$\Pgood\cap\gamma$ for all slabs $\gamma\in\Gamma(L)$.
\State For each $\gamma\in \Gamma(L)$, divide the points in $\Pgood\cap\gamma$
into groups of size $g$ (except for one leftover group of size $<g$), and
compute the upper hull of each with Graham's scan.
\State Recursively compute $\Bridges(P\cup Q,L)=\Bridges(\Pbad\cup (\Pgood\cup Q), L)$,
where $\Pgood\cup Q$ is partitioned into $\mu+O(n/g+b)$ groups of size $g$, each contained in a slab of $\Gamma(L)$.
\For{$\gamma\in\Gamma(L)$ with left side $\ell^-\in L$ and right side $\ell^+\in L$}
\State Let $v^-$ be the right endpoint of $\bridge(P,\ell^-)$.
\State Let $v^+$ be the left endpoint of $\bridge(P,\ell^+)$.
\State Let $P^{(\gamma)}=\{p\in P\cap\gamma\mid\text{$\IND(v^-)\leq\IND(p)\leq\IND(v^+)$ \& $x(v^-)\leq x(p)\leq x(v^+)$}\}$.
\State For every $Q_i\in\gamma$ let $Q^{(\gamma)}_i=\{q\in Q_i\cap\gamma\mid x(v^-)\leq x(q)\leq x(v^+)\}$. 
\State Let $L_0^{(\gamma)}$ be the lines of $L_0$ inside $\gamma$.
\State Recursively compute $\Bridges(P^{(\gamma)}\cup\{Q_i^{(\gamma)}\},L_0^{(\gamma)})$.
\EndFor
\addtocounter{linenumber}{-1}
\State Concatenate and return $\Bridges(P^{(\gamma)}\cup\{Q_i^{(\gamma)}\},L_0^{(\gamma)})$ and $\Bridges(P\cup Q,L)$.
\end{algorithmic}
\end{algorithm}

\begin{lemma}\label{lem:subproblem}
    Algorithm \ref{alg:prob2} computes $\Bridges(P_0,L_0)$ in time
    \[
    2^{O(\sqrt{\log\log\lambda})} \bigl(n + (\mu+n/g+\lambda g)\log\lambda\bigr)\log^2 g.
    \]
\end{lemma}
\begin{proof}
    Like \Cref{alg:prob1}, correctness follows from \Cref{lem:slabs-correct}; we recursively compute the bridges of only the safe subsets of $P_\gamma$ via the sets $P^{(\gamma)}$ and the vertices of the hulls of $\{Q_i^{(\gamma)}\}$.

    For the running time, applying \Cref{lem2} at Line 2 takes $O(n+\lambda)$ time. Computing the convex hulls for every group at Line 3 takes $O(g\log g)$ time per group via e.g.\ Graham's scan, and there are $O(n/g+b)$ such groups, so the total time is $O((n+bg)\log g)$. Line 4 is a recursive call with the recursive input $P$ consisting of only the left-over bad point pairs, and $Q$ being $\mu + O(n/g+b)$ groups of size $g$ that each are contained in a slab of $\Gamma(L)$ (refer to \Cref{fig:lambdaTree}). In the \textbf{for}-loop, points in $P\setminus \bigcup_\gamma P^{(\gamma)}$ are effectively pruned, by the same argument as in \Cref{cor:compute-slab-decomposition}.
    Letting $|\Pbad|=2z$, we thus have $\sum_\gamma |P^{(\gamma)}|\le n-z$.

    The running time thus satisfies the following recurrence for $\lambda>b\log b$:
    \begin{align*}
        T(n,\mu,\lambda)\ &\le\ \max_{\scriptsize\begin{array}{c}z,n_1,n_2,\ldots,\mu_1,\mu_2,\ldots:\\\sum_i n_i\le n-z\\\sum_i \mu_i\le \mu\end{array}}
    \Bigg( O((n+bg+\mu)\log g + \lambda)\\
    &\qquad\qquad\qquad+T(2z,\mu + O(n/g+b), O(b)) + \sum_i T(n_i,\mu_i,\lambda/b) \Bigg).
    \end{align*}
    For base case $\lambda=O(1)$, then $T(n,\mu,O(1))\le O((n+\mu)\log^2 g)$, since bridge-finding dualizes to 2D linear programming (LP), solvable in $O(k\log g\log\log g)\le O(k\log^2 g)$ time for $k$ convex $g$-gons \cite{CHAN1998147}.
    
    To bound $T(n,\mu,\lambda)$, we first choose $b=2$, for which the recurrence gives $T(n,\mu,\lambda) \le O((n+bg+\mu+\lambda)\log\lambda\log^2g)$.  In the following, we will adjust $b$ (and $\lambda$) recursively.    

Assume inductively that $T(n_i,\mu_i,\lambda) \le c_k \bigl(n_i\log^{1/k}\lambda +
(\mu_i+n_i/g+\lambda+bg)\log\lambda\bigr)\log^2 g$ for each $n_i,\mu_i$.
Then for $\lambda > b\log b$, we obtain
\begin{align*}
    T(n,\mu,\lambda)\ &\le\ \max_{\scriptsize\begin{array}{c}z,n_1,n_2,\ldots,\mu_1,\mu_2,\ldots:\\\sum_i n_i\le n-z\\\sum_i \mu_i\le \mu\end{array}}
  \Bigg( \sum_i T(n_i,\mu_i,\lambda/b) \\
  &\qquad\qquad\qquad+ O\bigl(c_k (n + z\log^{1/k}b + (\mu+n/g+\lambda+bg)\log b)\log^2 g\bigr) \Bigg).
\end{align*}
For base case $\lambda\le b\log b$, the inductive assumption becomes $T(n_i,\mu_i,\lambda)\le O\bigl(c_k (n_i\log^{1/k}b +
(\mu_i+n/g+\lambda+bg)\log\lambda)\log^2 g\bigr)$. %; since  
%the recurrence has depth $O(\frac{\log \lambda}{\log b})$ it solves to
%

%Since the recurrence has depth $O\!\left(\frac{\log \lambda}{\log b}\right)$, and each level contributes an additional $O(c_k\,n\log^2 g)$ term, this incurs a total extra $O\!\left(c_k\,n\frac{\log\lambda}{\log b}\log^2 g\right)$
%over the whole recurrence tree. On the other hand, the $z\log^{1/k} b$ term is not paid independently on every level: the contributions from all the $z$'s at internal nodes, together with the $n_i$'s at the leaves, charge each point in $P$ only $O(1)$ times in total, so these terms sum to only $O\!\left(c_k\,n\log^{1/k} b\,\log^2 g\right)$.
%Also, on any fixed level of the recurrence tree, the slab subproblems have disjoint line sets, so the sum of their $\lambda$-parameters is at most $\lambda$; in particular, there are only $O(\lambda/b)$ such subproblems on that level, and hence the total $bg$-charge on one level is $O(\lambda g)$. Likewise, since at each split we have $\sum_i \mu_i \le \mu$ and $\sum_i n_i \le n$, the total value of $\mu_i+n_i/g$ over all subproblems on any fixed level is at most $\mu+n/g$. Therefore all of these terms contribute only
%\[
%O\!\left(c_k(\mu+n/g+\lambda+\lambda g)\log\lambda\log^2 g\right)
%=O\!\left(c_k(\mu+n/g+\lambda g)\log\lambda\log^2 g\right)
%\]
%over the whole recurrence tree.

Since the recurrence has depth $O\!\left(\frac{\log \lambda}{\log b}\right)$, the per-level additive $O(c_k\,n\log^2 g)$ term contributes a total of $O\!\left(c_k\,n\frac{\log\lambda}{\log b}\log^2 g\right)$ over the whole recurrence tree. The $z\log^{1/k}b$ term does not incur such a depth factor: the contributions from all internal $z$'s together with the leaf $n_i$'s charge each point of $P$ only $O(1)$ times overall, and thus sum to $O\!\left(c_k\,n\log^{1/k}b\,\log^2 g\right)$.
Moreover, on each level the slab subproblems have disjoint line sets, so their $\lambda$-parameters sum to at most $\lambda$, and therefore their total $bg$-charge is $O(\lambda g)$. Likewise, since $\sum_i \mu_i\le \mu$ and $\sum_i n_i\le n$, the total value of $\mu_i+n_i/g$ on any fixed level is at most $\mu+n/g$. Hence these remaining terms contribute $O\!\left(c_k(\mu+n/g+\lambda g)\log\lambda\log^2 g\right)$
over the whole recurrence tree.
%Also, since at each split we have $\sum_i \mu_i \le \mu$ and $\sum_i n_i \le n$, the total value of $\mu_i+n_i/g$ over all subproblems on any fixed recursion level is at most $\mu+n/g$. Hence these terms contribute only $O\!\left(c_k(\mu+n/g)\log\lambda\log^2 g\right)$ over all $O\!\left(\frac{\log\lambda}{\log b}\right)$ levels.
%Likewise, the $\mu$-term and the other additive quantities inside the inductive bound are distributed among the children rather than duplicated, so their total on any fixed level remains bounded by the original instance and therefore sums to only $O\!\left(c_k(\mu+n/g+\lambda+bg)\log\lambda\log^2 g\right)$. 
Thus, overall, the recurrence solves to
\[
T(n,\mu,\lambda)\ \le\ O\left(c_k \left(n\tfrac{\log\lambda}{\log b} + n\log^{1/k}b + (\mu+n/g+\lambda g)\log\lambda\right)\log^2 g\right).
\]
Choosing $b$ with $\log b = \log^{k/(k+1)}\lambda$,
we get
\[
T(n,\mu,\lambda) \le O\bigl(c_k (n\log^{1/(k+1)}\lambda + (\mu+n/g+\lambda g)\log\lambda)\log^2 g\bigr).
\]
Thus, we can set $c_{k+1} = O(1)\cdot c_k$, i.e., $c_k=2^{O(k)}$, for any $k$. With $k=\sqrt{\log\log\lambda}$,
we obtain
%\[
%T(n,\mu,\lambda)\le 2^{O(k)} \bigl(n\log^{1/k}\lambda + (\mu+n/g+\lambda g)\log\lambda\bigr)\log^2 g
%\]
%for any $k$.
%Choosing $k=\sqrt{\log\log\lambda}$, we obtain
\[
T(n,\mu,\lambda)\le 2^{O(\sqrt{\log\log\lambda})} \bigl(n + (\mu+n/g+\lambda g)\log\lambda\bigr)\log^2 g.\qedhere
\]
\end{proof}

\begin{algorithm}
\caption{A randomised $n\cdot 2^{O(\sqrt{\log\log n})}$-time convex hull algorithm with a promise}\label{alg:prob3}
\begin{algorithmic}[1]
\Require Sequence $P$ of $n$ points, and fixed parameters $\lambda$ and $g$
\Ensure $\UH(P)$

\State Let $L$ be the vertical lines at a random sample of $\lambda$ $x$-coordinates of $P$. Sort $L$.
\State Compute $\Bridges(P,L)$ via \Cref{alg:prob2} (with $P$, $Q=\emptyset$, $\lambda$, and $g$).
\For{$\gamma\in\Gamma(L)$ with left side $\ell^-\in L$ and right side $\ell^+\in L$}
\State Let $v^-$ be the right endpoint of $\bridge(P,\ell^-)$.
\State Let $v^+$ be the left endpoint of $\bridge(P,\ell^+)$.
\State Let $P^{(\gamma)}=\{p\in P\cap\gamma\mid\text{$\IND(v^-)\leq\IND(p)\leq\IND(v^+)$ \& $x(v^-)\leq x(p)\leq x(v^+)$}\}$.
\State Compute $\UH(P^{(\gamma)})$ via Graham's scan.
\EndFor
\addtocounter{linenumber}{-1}
\State Concatenate $\{\UH(P^{(\gamma)})\}$ and $\Bridges(P,L)$ and return $\UH(P)$.
\end{algorithmic}
\end{algorithm}

\begin{theorem}\label{thm:prob}
    \Cref{alg:prob3} computes the upper hull of an input sequence $P$ that fulfils the convex hull promise in expected time $n\cdot2^{O(\sqrt{\log\log n})}$. % (which is $n\log^{o(1)}n$).
\end{theorem}
\begin{proof}
    W.h.p.\ $|P^{(\gamma)}|\le O((n/\lambda)\log n)$ for all $\gamma\in\Gamma(L)$.
    Thus, w.h.p.\ the total running time is $2^{O(\sqrt{\log\log \lambda})}(n+(n/g+\lambda g)\log \lambda)\log^2 g + O(n\log ((n/\lambda)\log n))$. Setting $\lambda=n/\log^2 n$ and $g=\log n$, we obtain a final running time of $2^{O(\sqrt{\log\log n})}n\cdot\log^2\log n = n2^{O(\sqrt{\log\log n})}$.
\end{proof}

\subparagraph*{Remarks.}
The $n\cdot 2^{O(\sqrt{\log\log n})}$ time complexity
(which is $o(n\log^\varepsilon n)$ for any $\varepsilon>0$) of our algorithm is unusual and highlights its intriguing ``double'' recursion.  We remark that similar bounds have appeared before in the computational geometry literature, 
for example, concerning data structures for dynamic planar convex hulls~\cite[Section~2]{Chan12} and word-RAM algorithms for three-dimensional convex hulls~\cite{ChanP07}.  Both of these previous works encounter a similar doubly recursive phenomenon, but the details are completely different.

Derandomization of our $n\cdot 2^{O(\sqrt{\log\log n})}$ algorithm is left as an open question.  The only place where randomization is used is in line~1 of \Cref{alg:prob3}; in particular, \Cref{alg:prob2} for computing $\Bridges(P,L)$ is already deterministic.
Our algorithm can also be made output-sensitive, with $n\cdot 2^{O(\sqrt{\log\log h})}$ running time for output size~$h$, and it can easily be adapted to the Pareto Front problem with the same running time.  See \Cref{app:pareto} for details.

%\subparagraph{Pareto fronts and other remarks.}

%Our algorithms can easily be adapted to the Pareto Front problem with the same running times; see \Cref{app:pareto} for details and other remarks. 

\section{Lower Bounds}
In previous sections, we showed that the ``promise'' of the convex hull vertices occurring in order in the input is sufficient to allow for $o(n \log n)$-time algorithms to construct the convex hull or Pareto front. 
A natural follow-up question is to ask to what degree the promise can be weakened, and still allow $o(n \log n)$-time algorithms. 
We show that there is essentially no wiggle-room between our result and the $\Omega(n \log n)$ running time barrier as we show an adversarial $\Omega(n \log n)$-time lower bound for when the promise is even slightly weakened.

\begin{definition}
    An input sequence $P$ is said to fulfil the \emph{almost convex hull promise} if there exists a unique point $p^* \in P$ such that all points $p \in \mathrm{conv}(P)$ with $p \neq p^*$ fulfil the convex hull promise.
    That is, for all $p, q \in \mathrm{conv}(P) - p^*$, $x(p) > x(q)$ if and only if $\IND(p) > \IND(q)$. The \emph{almost Pareto front promise} is defined analogously.
\end{definition}

\subparagraph*{Computational model.}
We show an adaptive adversarial, comparison-based lower bound. That is, we assume that the input sequence $P$ is given by an adversary that is aware of our convex hull construction program. Our program can compare values in $P$, but not manipulate these values. The adversary is able to observe the comparisons that we make, and is able to adapt the input sequence accordingly to either force our program to last longer or become incorrect. To allow for standard geometric and structural checks, we permit two types of \emph{comparisons}: coordinate comparisons and orientation tests. A more detailed description of the model of computation is given in \Cref{app:model}.

%\subsection{A lower bound for almost promises}

\subparagraph{Sketch of our lower bound argument.}
As input we take the union of two point sequences $P$ and $Q_\pi$ of sizes $n$ and $n-1$. $P$ consists of points on the parabola that are always on the upper convex hull. $Q_\pi$ depends on a permutation $\pi \colon [n - 1] \rightarrow [n - 1]$, and point $q_j \in Q_\pi$ lies just below the $\pi(j)$'s segment of $\UH(P)$. Additionally, we define point sets $Q^j_\pi$, which are equal to $Q_\pi$ except for the point $q_j$ that in $Q^j_\pi$ lies just above the $\pi(j)$'s segment of $\UH(P)$.

The lower bound argument now has two essential ingredients: first, because the adversary can switch the input from $Q_\pi$ to $Q^j_\pi$ without violating the almost promise, a correct algorithm must verify \emph{every} point in $Q_\pi$ as not on the convex hull. Second, the only way the algorithm can verify that $q_j$ is not on the convex hull is by testing its orientation against the unique edge of $P$ lying directly above it (the edge $\overline{p_{\pi(j)}p_{\pi(j)+1}}$). But if the algorithm has performed all of these comparisons, it has sufficient information to sort the unknown permutation $\pi$, which requires $\Omega(n \log n)$ comparisons. The complete argument is in \Cref{app:lowerconstruct}.

\begin{theorem}\label{thm:lower}
    For any deterministic correct algorithm that constructs the convex hull (resp. Pareto front) input sequences of size $2n$ fulfilling the almost promise, there exists at least one sequence requiring $\Omega(n \log n)$ comparisons.
\end{theorem}

\subparagraph{Negative answer to \Cref{Problem:supersequence}.}
Consider the point sequence $P \circ Q_\pi^j \circ P$. It contains a subsequence with the vertices of the convex hull in order, since we can take the points left of $q_j$ from the first $P$ and the points right of $q_j$ from the second $P$. Thus, it fulfils the conditions of \Cref{Problem:supersequence}. But since the convex hull of $P \circ Q_\pi^j \circ P$ equals the convex hull of $P \circ Q_\pi^j$, the lower bound of \Cref{thm:lower} applies.

\bibliographystyle{plainurl}
\bibliography{bibliography}

\appendix

\section{Output-Sensitive Algorithms and Pareto Front Algorithms}\label{app:pareto}

We observe that \Cref{alg:prob3} can be adapted to achieve output-sensitive running time:

\begin{corollary}
    \Cref{alg:prob3} can be modified to compute the upper hull of an input sequence $P$ that fulfils the convex hull promise in expected time $n\cdot2^{O(\sqrt{\log\log h})}$, where $h$ is the size of the convex hull.
\end{corollary}
\begin{proof}
Consider a guess $h$ of the size of the convex hull. We modify \Cref{alg:prob3} by replacing Graham's scan with Jarvis march in Line 7. The expected run time is 
$2^{O(\sqrt{\log\log b})}(n+(n/g+b)\log b)\log^2 g + O(((n/b)\log b)\cdot h)$,
which is $n2^{O(\sqrt{\log\log h})}$, by setting $b=h^2$ and $g=\log h$.

Since we do not know $h$ in advance, we employ the standard guessing trick, trying $h=2^{2^i}$ until $h>n^{0.01}$, in which case we default to the Graham's scan version which is then already a $n\cdot2^{O(\sqrt{\log\log h})}$ time algorithm.
Since the guesses are $h_i=2^{2^i}$, the total time over all unsuccessful guesses and the final successful one is
$\sum_{i=1}^{t} n\cdot 2^{O(\sqrt{\log\log h_i})}=n\cdot 2^{O(\sqrt{t})}$,
where $t$ is the first index with $h_t\ge h$. As $t=\Theta(\log\log h)$, this is $n\cdot 2^{O(\sqrt{\log\log h})}$.
\end{proof}

Our algorithms can also be adapted to solve the Pareto front problem.
When computing the Pareto front of a point sequence $P$, we can define $\bridge(P,\ell)$, for any vertical line $\ell$, as the unique edge of the Pareto front that intersects $\ell$.
Observe that when computing the Pareto front, $\bridge(P,\ell)$ is defined by the Pareto front vertex immediately before and after the line $\ell$. In particular, this can be computed in $O(|P|)$ time via a linear scan of $P$. Therefore, \Cref{lem:abstract-bridge-compute} also holds for the Pareto front. With this, we obtain \Cref{cor:pareto1} as Pareto-front version of \Cref{thm:det}.

\begin{corollary}\label{cor:pareto1}
    \Cref{alg:prob1} can be modified to compute the Pareto front of a sequence of points fulfilling the Pareto-front promise in $O(n\sqrt{\log n})$ time.
\end{corollary}

For Algorithm~\ref{alg:prob2}, we observe that all lines apply to both the convex hull and the Pareto front.
Specifically, we note that to compute the bridge across a line for which we have compressed the groups (in Line 3) in its neighbouring slabs, it suffices to maintain only the first and last point of the Pareto front of each group. 
(In fact, we could even avoid dividing into groups altogether, so many of the steps could be simplified for the Pareto front problem.)
This recovers a similar bound on $T(n,\mu,O(1))$ in the proof of \Cref{lem:subproblem}, resulting in a Pareto-front version of \Cref{lem:subproblem}. So we obtain \Cref{cor:pareto2} as Pareto-front version of \Cref{thm:prob}.

\begin{corollary}\label{cor:pareto2}
    \Cref{alg:prob3} can be modified to compute the Pareto front of a sequence of points fulfilling the Pareto-front promise in expected time  $O(n\cdot 2^{O(\sqrt{\log\log n})})$ 
    (or $n\cdot2^{O(\sqrt{\log\log h})}$ where $h$ is the number of points on the front).
\end{corollary}

\section{Additional Details for Lower Bounds}\label{app:lower}

\subsection{Model of Computation}\label{app:model}

Formally, we assume that our input is a set $P$ of $n$ points given as sequence. 
An \emph{algorithm} specifies for each fixed input size $n$, a (RAM) programme which is a list of \emph{instructions}. Instructions can manipulate memory, halt the programme (marking its termination) or make \emph{comparisons}. 
After each instruction, the programme continues to the next instruction in the list, except for when a \emph{comparison} instruction is executed. 
To allow for standard geometric and structural checks, we permit two types of \emph{comparisons}: coordinate comparisons and orientation tests. A \emph{coordinate comparison} compares the $x$-coordinate or $y$-coordinate of two points $p_i$ and $p_j$. An \emph{orientation test} evaluates whether or not a point $p_k$ lies strictly above the line through two points $p_i$ and $p_j$. A comparison returns a Boolean, and if the return value is true,
%A \emph{comparison} specifies integers $i, j \in [n]$ and an integer $\ell$. It then compares the points $p_i$ and $p_j$, outputting a Boolean (e.g., by comparing their $x$-coordinate or $y$-coordinate). 
%If the Boolean is true, 
the programme then executes the $\ell$'th instruction, where $\ell$ is given by the algorithm. If the Boolean is false, it continues to the next instruction. 
We assume a model of computation where our programme cannot have instructions that manipulate $P$. I.e., all non-comparison instructions must consist of inputs that are defined independently of the input $P$ (such as deterministic, input-independent constants defined to index memory).
We observe that this model of computation supports our upper bound algorithms.

\subsection{Lower Bound Construction}\label{app:lowerconstruct}

\subparagraph*{Defining the input sequence.}
The base of our construction is an input sequence which, for notational convenience, is the union of two input sequences $P$ and $Q_\pi$.
The point set $P$ consists of $n$ points from a parabola. Specifically, we define each point $p_i$ as $p_i := (i, n^2 - i^2)$. Note that the convex hull and Pareto front of $P$ are exactly $P$, in-order. 
The point set $Q_\pi$ consists of $n-1$ points that are selected by the adversary, who constructs these by selecting a permutation $\pi : [n - 1] \rightarrow [n - 1]$. 
Each point $q_j \in Q_\pi$ is then placed just below the $\pi(j)$'th edge of the convex hull of $P$. 
That is, $q_j := \left( \pi(j) + \frac{1}{2}, n^2 - \frac{(\pi(j) + 1)^2 + \pi(j)^2}{2} -\varepsilon\right)$ for some sufficiently small $\varepsilon > 0$. Thus, the convex hull of the input sequence $P \circ Q_\pi$ equals $P$, in-order. 
Note that the input sequence $P \circ Q_\pi$ fulfils the promise. Thus, if we are given an input sequence where we \emph{know} that it is $P \circ Q_\pi$ for some unknown permutation $\pi$ then we can construct the convex hull for $P \circ Q_\pi$ in $o(n \log n)$ time.
However, we show that if the adversary can every so slightly modify this input to only fulfil an almost promise, then a correct algorithm requires $\Omega(n \log n)$ time.

\subparagraph{Adversarial almost promises.}
Finally, we define alternative point sets for the adversary to invoke. For any fixed permutation $\pi$ and integer $j \in [n]$, we define the point set $Q^j_\pi$. The point set $Q^n_\pi = Q_\pi$, and each other $Q^j_\pi$ is $Q_\pi$ with one change: the point $q_j$ is changed to $q_j := (\pi(j) + \frac{1}{2}, n^2 - (\pi(j) + \frac{1}{2})^2 )$. The point $q_j$ thus appears on the convex hull (or Pareto front) of $P \circ Q^j_\pi$ in-between $p_{\pi(j)}$ and $p_{\pi(j) + 1}$. 
Observe that each input sequence $P \circ Q^j_\pi$ fulfils the almost promise. 

\begin{proof}[Proof of \Cref{thm:lower}]
    We consider any fixed deterministic algorithm $\mathcal{A}$ to construct the convex hull of our input sequence $I$, and observe that given the fact that our input sequence is $P \circ Q_\pi^j$ for some $(\pi, j)$, it must be that the upper convex hull of $I$ equals its convex hull.
    Note that a point $q \in I$ does not appear on the upper convex hull, if and only if there exist distinct points $a, b \in I$ such that $q$ lies strictly under the line through $a$ and $b$. 
    If the algorithm $\mathcal{A}$ finds a triple $(a, b, q)$ where $q$ lies under the line $\overline{ab}$ through $a$ and $b$, and, the algorithm compares $\overline{ab}$ to $q$ through an orientation test, then we say that the algorithm $\mathcal{A}$ has created a \emph{witness} for $q$.

    The adaptive adversary will observe the comparisons made by our programme, and will choose at any point in time the input parameters $i$ and $\pi$ such that the input $P \circ Q_\pi^j$ is consistent with the comparisons made thus far, such that $\mathcal{A}$ is either incorrect or maximises its running time. 
    Note that for any input $P \circ Q_\pi^n$, each point $q_j \in Q_\pi^n = Q_\pi$ has a unique witness which is $(\pi(j), \pi(j) + 1, q)$. Indeed, by choosing $\varepsilon$ (see above) sufficiently small, the pair $(\pi(j), \pi(j)+1)$ is the unique pair of points such that $q_j$ lies under their corresponding line. 
    It follows that any algorithm $\mathcal{A}$ that constructs for all $q_j \in Q_\pi^n$ their corresponding witness $(\pi(j), \pi(j) + 1, q)$ recovers $\pi$. In other words, any algorithm that identifies for each $q_j \in Q_\pi^n$ the unique edge of $\UH(P)$ that lies above $q_j$, sorts $Q_\pi^n$ along $\pi$.  
    We first claim that any fixed algorithm $\mathcal{A}$ for which there exists an input $P \circ Q^n_\pi$ where $\mathcal{A}$ \emph{does not} construct for each $q_j \in Q_\pi^n$ its witness, is incorrect. 

    Indeed, consider such an algorithm and let there exist at least one $q_j \in Q_\pi^n$ for which the algorithm did not compare $q_j$ to the line through $p_{\pi(j)}$ and $p_{\pi(j) + 1}$ via an orientation test.
    The adversary now provides us with one of two inputs: either $Q_\pi^j$ or $Q_\pi^n$. 
    In the first input sequence, $q_j$ appears in the output sequence in-between $p_{\pi(j)}$ and $p_{\pi(j) + 1}$.
    In the second input sequence, $q_j$ does not appear in the output sequence. 
    However, the point $q_j$ lies above the line $\overline{ab}$ for all $(a, b) \neq (\pi(j), \pi(j) + 1)$. Thus, the algorithm $\mathcal{A}$ has no way to distinguish between these two inputs. The algorithm must decide whether $q_j$ is part of the output, and it follows that on at least one of these inputs, the algorithm must be incorrect. 

    Constructing for each $q_j \in Q^n_\pi$ its witness equals sorting $Q_\pi^n$ along the permutation $\pi$. Since sorting has an $\Omega(n \log n)$ comparison-based lower bound, it follows that there exists a permutation $\pi$ for which a correct algorithm $\mathcal{A}$ requires $\Omega(n \log n)$ comparisons which concludes the proof. 
\end{proof}

\end{document}